\pdfoutput=1
\RequirePackage{ifpdf}
\ifpdf 
\documentclass[pdftex]{sigma}
\else
\documentclass{sigma}
\fi

\usepackage[mathscr]{eucal}

\newcommand{\rvac}{|0 \rangle}
\newcommand{\Z}{\mathbb Z}
\newcommand{\C}{\mathbb C}

\numberwithin{equation}{section}

\newtheorem{Theorem}{Theorem}[section]
\newtheorem{Corollary}[Theorem]{Corollary}
\newtheorem{Lemma}[Theorem]{Lemma}
\newtheorem{Proposition}[Theorem]{Proposition}
 { \theoremstyle{definition}

 }

\begin{document}
\allowdisplaybreaks

\newcommand{\arXivNumber}{2106.04773}

\renewcommand{\PaperNumber}{089}

\FirstPageHeading

\ShortArticleName{Virasoro Action on the $Q$-Functions}

\ArticleName{Virasoro Action on the $\boldsymbol{Q}$-Functions}

\Author{Kazuya AOKAGE~$^{\rm a}$, Eriko SHINKAWA~$^{\rm b}$ and Hiro-Fumi YAMADA~$^{\rm c}$}

\AuthorNameForHeading{K.~Aokage, E.~Shinkawa and H.-F.~Yamada}

\Address{$^{\rm a)}$~Department of Mathematics, National Institute of Technology, Ariake College,\\
\hphantom{$^{\rm a)}$}~Fukuoka 836-8585, Japan}
\EmailD{\href{mailto:aokage@ariake-nct.ac.jp}{aokage@ariake-nct.ac.jp}}

\Address{$^{\rm b)}$~Advanced Institute for Materials Research, Tohoku University, Sendai 980-8577, Japan}
\EmailD{\href{mailto:eriko.shinkawa.e8@tohoku.ac.jp}{eriko.shinkawa.e8@tohoku.ac.jp}}

\Address{$^{\rm c)}$~Department of Mathematics, Kumamoto University, Kumamoto 860-8555, Japan}
\EmailD{\href{mailto:hfyamada@kumamoto-u.ac.jp}{hfyamada@kumamoto-u.ac.jp}}

\ArticleDates{Received June 10, 2021, in final form October 05, 2021; Published online October 08, 2021}

\Abstract{A formula for Schur $Q$-functions is presented which describes the action of the Virasoro operators. For a strict partition, we prove a concise formula for $L_{-k}Q_{\lambda}$, where $L_{-k}$ $(k\geq 1)$ is the Virasoro operator.}


\Classification{17B68; 05E10}

\rightline{\it To Minoru Wakimoto on his 80th birthday}

\section{Introduction}

The aim of this paper is to discuss Schur $Q$-functions in connection with a representation of the Virasoro algebra.
Schur $Q$-functions are labelled by strict partitions and are defined as the Pfaffian of an alternating matrix.
Let $A=(a_{ij})_{1\leq i,j\leq 2m}$ be an alternating $2m \times 2m$ matrix.
The Pfaffian of $A$ is
\begin{gather*}
\operatorname{Pf}(A) := \sum_{\sigma \in F_{2m}}({\rm{sgn}} \, \sigma)a_{\sigma(1)\sigma(2)}a_{\sigma(3)\sigma(4)}\cdots
a_{\sigma(2m-1)\sigma(2m)},
\end{gather*}
where
\begin{gather*}
F_{2m} := \{\sigma \in S_{2m}; \sigma(1) < \sigma(3) < \cdots <\sigma(2m-1), \sigma(i)<\sigma(i+1)\, (i=1,3,\dots, 2m-1)\}.
\end{gather*}
We see that $|F_{2m}| = (2m-1)!!$. The Laplace expansion of $\operatorname{Pf}(A)$ is as follows.
For $1 \leq i_{1} < \cdots <i_{2\ell}\leq 2m$, let $A_{i_{1}i_{2}\dots i_{2\ell}}$ be the $2\ell \times 2\ell$
alternating matrix consisting of $i_{1}$th row, $i_{2}$th row, \dots, and $i_{1}$th column, $i_{2}$th column, \dots.
Then
\begin{gather*}
\operatorname{Pf}(A) = \sum_{i=2}^{2m}(-1)^{i} \operatorname{Pf}(A_{1i}) \operatorname{Pf}(A_{2 \dots \widehat{i} \dots 2m}).
\end{gather*}
Here $\widehat{i}$ means the omission of $i$. We will utilize this quadratic relation to derive
formulas for $Q$-functions.

Our previous paper \cite{A-S-Y} gives a formula of $L_{k}Q_{\lambda}$ for $k\geq 1$,
where $L_{k}$ denotes the Virasoro operator.
As a continuation of~\cite{A-S-Y} we give in the present paper a formula for $L_{-k}Q_{\lambda}$.
Section~\ref{section2} is a review of $Q$-functions containing some identities which do not seem to be obviously
derived from Pfaffian identities.
In Section~\ref{section3} we first recall the reduced Fock representation of the Virasoro algebra on the space of the $Q$-functions.
Then the main result is given. Proofs consist of direct, simple calculations.

The Virasoro representations of this paper may be applied to, for example, the Kontsevich matrix models by certain rescaling. However we will not discuss here any relationship. Our motivation is to clarify the representation theoretical nature of the Hirota equations for certain soliton type hierarchies. In the final section we will give a conjecture on the Hirota equations for the KdV hierarchy.

\section[Schur's Q-functions]{Schur's $\boldsymbol{Q}$-functions}\label{section2}

A partition is an integer sequence
$\lambda=(\lambda_1,\lambda_2,\dots,\lambda_{\ell})$, $\lambda_1\geq \lambda_2 \geq \dots \geq \lambda_{\ell}> 0$, whose size is $|\lambda|=\lambda_1+\lambda_2+\dots+\lambda_{\ell}$. The number of nonzero parts is the length of $\lambda$, denoted by $\ell(\lambda)$. Let~${\mathcal{SP}}(n)$ be the set of partitions of $n$ into distinct parts. We call a $\lambda \in {\mathcal{SP}(n)}$ strict partition of~$n$.
Let $V={\C}[t_j; j\geq 1,{\rm odd}]$. This is decomposed as $V=\bigoplus_{n=0}^{\infty} V(n)$, where $V(n)$ is the space of homogeneous polynomials of degree $n$, according to the counting $\deg t_j=j$. An inner product of~$V$ is defined by $\langle F, G \rangle=F\big(2\widetilde{\partial}\big)\overline{G}(t)|_{t=0}$, where $2\widetilde{\partial}=\big(2\partial_1, \frac{2}{3}\partial_3, \frac{2}{5}\partial_5, \dots\big)$ with $\partial_{j}=\frac{\partial}{\partial t_j}$.

Schur's $Q$-functions are defined in our context as follows. Put $\xi(t,u)=\sum_{j\geq 1,\,{\rm odd}}t_ju^j$ and define $q_n(t)\in V(n)$ by
\begin{gather*}
{\rm e}^{\xi(t,u)}=\sum_{n=0}^{\infty}q_n(t)u^n.
\end{gather*}
For integers $a$, $b$ with $a>b>0$, define
\begin{gather*}
Q_{a b}(t):=q_{a}(t)q_{b}(t) + 2\sum_{i=1}^{b}(-1)^{i}q_{a+i}(t)q_{b-i}(t),
\\
Q_{b a}(t):=-Q_{a b}(t).
\end{gather*}
Finally, the $Q$-function labelled by the strict partition $\lambda=(\lambda_1,\lambda_2,\dots, \lambda_{2m})$ $(\lambda_1 >\lambda_2 >\dots >\lambda_{2m} \geq 0)$ is defined by
\begin{gather*}
Q_{\lambda}(t)=Q_{\lambda_1\lambda_2\dots \lambda_{2m}}(t)=\operatorname{Pf}{(Q_{\lambda_i\lambda_j})}_{1\leq i,j \leq 2m}.
\end{gather*}
The $Q$-function $Q_{\lambda}(t)$ is homogeneous of degree $|\lambda|$. It is known that $\{Q_{\lambda}(t);\, |\lambda|=n\}$ forms an orthogonal basis for~$V(n)$, with respect to the above inner product. As Pfaffians, they satisfy the quadratic relations (cf.~\cite{H-H}):

If $\ell(\lambda)$ is odd,
\begin{gather*}
 Q_{\lambda_1 \lambda_2 \dots \lambda_{\ell}}(t) = \sum^{\ell}_{i=1}(-1)^{i+1}q_{\lambda_i}(t)Q_{\lambda_1 \dots \widehat{\lambda_i} \dots \lambda_{\ell}}(t).
\end{gather*}

If $\ell(\lambda)$ is even,
\begin{gather*}
 Q_{\lambda_1 \lambda_2 \dots \lambda_{\ell}}(t)=\sum^{\ell}_{i=2}(-1)^{i}Q_{\lambda_1\lambda_i}(t)Q_{\lambda_2 \dots \widehat{\lambda_i}\dots \lambda_{\ell}}(t).
\end{gather*}

It is convenient to define $Q$-function $Q_{\alpha}(t)$ for any non-negative integer sequence $\alpha=(\alpha_1,\alpha_2,\dots,\alpha_{\ell})$. We adopt the following rule for permutations of indices:
\begin{enumerate}\itemsep=0pt
\item If $\alpha_1,\alpha_2,\dots,\alpha_{\ell}$ are all distinct, then $\sigma(\alpha)$ is a strict partition for some permutation $\sigma\in S_{\ell}$, and
\begin{gather*}
Q_{\alpha}(t)=({\rm sgn}\, \sigma)Q_{\sigma(\alpha)}(t).
\end{gather*}
\item If $\alpha_{i}=\alpha_{j}>0$ for some $i\ne j$, then $Q_{\alpha}(t)=0$.

\item Using permutations, $0$'s should be moved in the tail of $\alpha$, keeping $0$'s order. After such permutation, all $0$'s should be deleted.
\end{enumerate}
For example, we have $Q_{0,2,3,0,1}(t)=-Q_{3,2,1}(t)$.
Detailed arguments are found in \cite[Theo\-rem~9.2]{H-H}.
Note that the above quadratic relations hold for $Q_{\alpha}(t)$ with non-negative integer sequence $\alpha=(\alpha_1,\alpha_2,\dots,\alpha_{\ell})$. We also agree that, for $a>0$, $Q_{a,-a}(t)={(-1)}^{a-1}$.

\begin{Lemma}\label{Q-relation2}
Let $\alpha=(\alpha_1,\alpha_2,\dots,\alpha_{\ell})$ be a non-negative integer sequence, and $x$, $y$ be non-negative integers.
\begin{enumerate}\itemsep=0pt
\item[$(1)$] If $\ell(\alpha)$ is\ odd,
\begin{gather*}
Q_{\alpha x}=-q_{x}Q_{\alpha}-\sum_{i=1}^{\ell}{(-1)}^{i}q_{\alpha_i}Q_{\alpha_1 \dots \widehat{\alpha_i} \dots \alpha_{\ell} x}-\sum_{i=1}^{\ell}{(-1)}^{i}Q_{\alpha_i x}Q_{\alpha_1 \dots \widehat{\alpha_{i}} \dots \alpha_{\ell}},
\end{gather*}
\item[$(2)$] If $\ell(\alpha)$ is even,
\begin{gather*}
Q_{\alpha x}=-q_{x}Q_{\alpha}+\sum_{i=2}^{\ell}{(-1)}^{i}Q_{\alpha_1\alpha_i}Q_{\alpha_2\dots \widehat{\alpha_i} \dots\alpha_{\ell} x}+\sum_{i=2}^{\ell}{(-1)}^{i}Q_{\alpha_1\alpha_i x}Q_{\alpha_2\dots\widehat{\alpha_{i}}\dots\alpha_{\ell}},
\end{gather*}
\item[$(3)$] If $\ell(\alpha)$ is odd,
\begin{gather*}
Q_{\alpha x y}=-Q_{x y}Q_{\alpha}-\sum_{i=1}^{\ell}{(-1)}^{i}q_{\alpha_i}Q_{\alpha_1 \dots \widehat{\alpha_i} \dots \alpha_{\ell} x y}-\sum_{i=1}^{\ell}{(-1)}^{i}Q_{\alpha_i x y}Q_{\alpha_1 \dots \widehat{\alpha_{i}} \dots \alpha_{\ell}},
\end{gather*}
\item[$(4)$] If $\ell(\alpha)$ is even,
\begin{gather*}
Q_{\alpha x y}=-Q_{x y}Q_{\alpha}+\sum_{i=2}^{\ell}{(-1)}^{i}Q_{\alpha_{1} \alpha_i}Q_{\alpha_2 \dots \widehat{\alpha_i}
\dots \alpha_{\ell} x y}+\sum_{i=2}^{\ell}{(-1)}^{i}Q_{\alpha_1 \alpha_i x y}Q_{\alpha_2 \dots \widehat{\alpha_{i}} \dots \alpha_{\ell}}.
\end{gather*}
\end{enumerate}
\end{Lemma}
\begin{proof}
Let $\ell(\alpha,x,y)$ be an even number. From the Pfaffian identity for $Q_{\alpha x y}$, the case (4) follows easily.
The cases (3) and (2) follow from (4) by setting $\alpha_1=0$ and $y=0$, respectively. Finally case (1) is obtained from case (3) by setting $y=0$.
\end{proof}

Next, we recall the boson-fermion correspondence for neutral free fermions $\phi_i\ (i\in \Z)$ (cf.~\cite{J-M}). The Clifford algebra $\mathfrak{B}$ is generated by free fermions $\phi_{i}\ (i\in \Z)$ satisfying the anti-commutation relation:
\begin{gather*}
[\phi_i,\phi_j]_{+}={(-1)}^{i}\delta_{i,-j}.
\end{gather*}
The vector space $F_{B}$ has a basis consisting of $\phi_{i_1}\phi_{i_2}\cdots\phi_{i_s} \rvac$, $i_1>i_2>\dots >i_s\geq 0$, where $\rvac$ is the vacuum vector.
The Clifford algebra $\mathfrak{B}$ acts on $F_{B}$ by $\phi_i\rvac =0$, $i<0$. For odd number~$n$, we define the Hamiltonian by
\begin{gather*}
H^{B}_{n}=\frac{1}{2}\sum_{i \in \Z}{(-1)}^{i-1}\phi_{i}\phi_{-n-i}.
\end{gather*}
The operators $H_{n}^{B}$ $(n\in {\rm \Z_{odd}})$ generate a Heisenberg algebra $\mathfrak{H}^{B}$ with $\big[H_{n}^{B},H_{m}^{B}\big]=\frac{n}{2}\delta_{n,-m}$.
It is known that $F_{B}$ is isomorphic to~$V$:
\begin{gather*}
\sigma_{B}(\rvac)=1,\qquad \sigma_{B}\big(H_{n}^{B}\rvac\big):=\frac{\partial}{\partial p_n}, \qquad \sigma_{B}\big(H_{-n}^{B}\rvac\big):=np_{n},\qquad n\geq 1,\ {\rm odd}.
\end{gather*}
The map $\sigma_{B}\colon \mathfrak{H}^{B} \longrightarrow V$ is called the boson-fermion correspondence of type~$B$. Moreover, for the basis of $F_B$
\[
\sigma_{B}\left(\phi_{\lambda_1}\phi_{\lambda_2}\cdots \phi_{\lambda_{\ell}}\rvac\right)=
 \begin{cases}
2^{-\frac{\ell}{2}}Q_{\lambda_1 \lambda_2 \dots \lambda_{\ell}}& \text{if $\ell$ is even}, \\
2^{-\frac{\ell+1}{2}}Q_{\lambda_1 \lambda_2 \dots \lambda_{\ell}}& \text{if $\ell$ is odd}.
\end{cases}
\]
In what follows, we denote $\sigma_{B}\left(\phi_{\lambda_1}\phi_{\lambda_2}\cdots \phi_{\lambda_{\ell}}\rvac\right)$ by $\phi_{\lambda_1}\phi_{\lambda_2}\cdots \phi_{\lambda_{\ell}}\rvac$.

\begin{Proposition}\label{Q-relation1}
Let $n=2m$. Then
\begin{gather*}
\sum_{i=0}^{m-1}(2i+1)t_{2i+1}(n-(2i+1))t_{n-(2i+1)}=2\sum_{i=0}^{m-1}{(-1)}^{i}(m-i)Q_{n-i, i}.
\end{gather*}
\end{Proposition}
\begin{proof}
First we rewrite the left-hand side of this equation by using power sum symmetric functions.
\begin{gather*}
\sum_{i=0}^{m-1}(2i+1)t_{2i+1}(n-(2i+1))t_{n-(2i+1)}=4\sum_{i\geq 1, \, {\rm odd}}^{n}p_{i}p_{n-i}.
\end{gather*}
For an odd number $i$, the operator $H_{i}^{B}$ acts on ${F_{B}}$. By the boson-fermion correspondence, the right hand side equals
\begin{gather*}
 \sum_{i\geq 1, \, {\rm odd}}^{n}\sum_{j,k \in \Z}{(-1)}^{j+k}\phi_{j}\phi_{-j+i}\phi_{k}\phi_{-k+(n-i)} \rvac.
\end{gather*}
Since free fermions $\phi_{i}$ $(i<0)$ act on vacuum vector $\rvac$ as~$0$, the above summation becomes
\begin{gather}
 \sum_{i\geq 1,\, {\rm odd}}^{n}\sum_{-n+i \leq j \leq n \atop 0\leq k\leq n-i}{(-1)}^{j+k}\phi_{j}\phi_{-j+i}\phi_{k}\phi_{-k+(n-i)} \rvac \nonumber\\
\qquad{} =\sum_{i\geq 1, \,{\rm odd}}^{n}\sum_{-n+i \leq j < 0 \atop 0\leq k\leq n-i}{(-1)}^{j+k}\phi_{j}\phi_{-j+i}\phi_{k}\phi_{-k+(n-i)} \rvac \label{ppQ1}
\\
\qquad\quad{} +\sum_{i\geq 1, \,{\rm odd}}^{n}\sum_{0 \leq j \leq i \atop 0\leq k\leq n-i}{(-1)}^{j+k}\phi_{j}\phi_{-j+i}\phi_{k}\phi_{-k+(n-i)} \rvac \label{ppQ2}
\\
\qquad\quad{} +\sum_{i\geq 1,\, {\rm odd}}^{n}\sum_{i < j \leq n \atop 0\leq k\leq n-i}{(-1)}^{j+k}\phi_{j}\phi_{-j+i}\phi_{k}\phi_{-k+(n-i)} \rvac.\label{ppQ3}
\end{gather}
Here it is verified that the part~(\ref{ppQ2}) equals~$0$. Next, we consider the parts (\ref{ppQ1}) and (\ref{ppQ3}). For the term ${(-1)}^{j+k}\phi_{j}\phi_{-j+i}\phi_{k}\phi_{-k+(n-i)}$, we only need to consider the cases that
$j$ or $-j+i$ belongs to $\{-k, -(n-i)+k\}$. That is, ${(-1)}^{j+k}\phi_{j}\phi_{-j+i}\phi_{k}\phi_{-k+(n-i)}$
\begin{gather*}
= \begin{cases}
 \phi_{-k}\phi_{i+k}\phi_{k}\phi_{-k+(n-i)} & \text{if $j=-k$}, \\
 -\phi_{i+k}\phi_{-k}\phi_{k}\phi_{-k+(n-i)} & \text{if $j=i+k$}, \\
 -\phi_{-(n-i)+k}\phi_{-j+i}\phi_{k}\phi_{-k+(n-i)} & \text{if $j=-(n-i)+k$},\\
 \phi_{n-k}\phi_{-j+i}\phi_{k}\phi_{-k+(n-i)} & \text{if $j=n-k$},
\end{cases}
\\
=
 \begin{cases}
 {(-1)}^{k+1}\phi_{i+k}\phi_{-k+(n-i)} & \text{if $j=-k$ or $j=i+k$}, \\
 {(-1)}^{k+1}\phi_{k}\phi_{-k+n} & \text{if $j=-(n-i)+k$ or $j=n-k$}.
\end{cases}
\end{gather*}
Hence we have
\begin{gather*}
\sum_{i\geq 1, \, {\rm odd}}^{n}\sum_{-n+i \leq j \leq n \atop 0\leq k\leq n-i}{(-1)}^{j+k}\phi_{j}\phi_{-j+i}\phi_{k}\phi_{-k+(n-i)} \rvac\\
\qquad{}=\sum_{i\geq 1,\, {\rm odd}}^{n}\sum_{0\leq k\leq n-i}{(-1)}^{k+1}\left(2\phi_{i+k}\phi_{-k+(n-i)}+2\phi_{k}\phi_{-k+n}\right) \rvac\\
\qquad{}=2\sum_{i\geq 1,\, {\rm odd}}^{n}\sum_{0\leq k\leq n-i}{(-1)}^{k+1}\left(\phi_{i+k}\phi_{-k+(n-i)}+\phi_{k}\phi_{-k+n}\right) \rvac\\
\qquad{}=2\sum_{i=0}^{m-1}{(-1)}^{i}(n-2i)\phi_{n-i}\phi_{i} \rvac\\
\qquad{}=2\sum_{i=0}^{m-1}{(-1)}^{i}(m-i)Q_{n-i,i}.\tag*{\qed}
\end{gather*}\renewcommand{\qed}{}
\end{proof}

\section{Reduced Fock representation of the Virasoro algebra}\label{section3}
For a positive odd integer $j$, put $a_{j}=\sqrt{2}\partial_j$ and $a_{-j}=\frac{j}{\sqrt{2}}t_j$ so that they satisfy the Heisenberg relation as operators on $V$:
\begin{gather*}
[a_{j}, a_{i}] = j\delta_{j+i,0}.
\end{gather*}
For an integer $k$, put
\begin{gather*}
L_k = \frac{1}{2} \sum_{j \in \textrm{$\Z$}_{\rm{odd}}} {:}a_{-j} a_{j+2k}{:} +\frac{1}{8}\delta_{k,0},
\end{gather*}
where
\begin{gather*}
 {:}a_{j} a_{i}{:}
=
\begin{cases}
 a_j a_i &\text{if $j\leq i$}, \\
 a_ia_j &\text{if $j> i$}
\end{cases}
\end{gather*}
is the normal ordering. For example,
\begin{gather*}
L_2=2\partial_1\partial_3+\sum_{j\geq1,\,{\rm odd}}jt_j\partial_{j+4},\qquad L_1=\partial_1^2+\sum_{j\geq 1,\,{\rm odd}}jt_{j}\partial_{j+2},
\\
L_{0}=\sum_{j\geq 1}jt_j\partial_{j}+\frac{1}{8}\,{\rm id},
\\
L_{-1}=\frac{1}{4}t_{1}^{2}+\sum_{j\geq 3,\,{\rm odd}}jt_{j}\partial_{j-2}, \qquad \text{and} \qquad L_{-2}=\frac{3}{2}t_1t_3+\sum_{j\geq 5,\,{\rm odd}}jt_j\partial_{j-4}.
\end{gather*}

More generally, it is verified, by Proposition~\ref{Q-relation1}, that
\begin{gather*}
L_{-k}=\frac{1}{4}\sum_{i=0}^{k-1}(2i+1)t_{2i+1}(2k-(2i+1))t_{2k-(2i+1)}+\sum_{i\geq 2k+1,\,{\rm odd}}it_i\partial_{i-2k}\\
\hphantom{L_{-k}}{} =\frac{1}{2}\sum_{j=0}^{k-1}{(-1)}^{j}(k-j)Q_{2k-j,j}+\sum_{j\geq 2k+1,\,{\rm odd}}jt_j\partial_{j-2k}.
\end{gather*}

It is known that the operators $L_k$ on $V$ satisfy the Virasoro relation:
\begin{gather*}
[L_k, L_\ell] = 2(k-\ell)L_{k+\ell}+\frac{k^3-k}{3}\delta_{k+\ell, 0},\qquad k, \ell \in \Z.
\end{gather*}
A representation of the Virasoro algebra ${\mathcal L}=\oplus_{k\in \mathbb{Z}}{\mathbb{C}\ell_{k}}\oplus{\mathbb{C}z}$ with central charge $1$ is given by
$\ell_{k} \mapsto \frac{1}{2}L_{k}$, $z \mapsto 1$, which we recall the reduced Fock representation. We have $L_k \cdot v \in V(n-2k)$ for $v \in V(n)$. The inner product $\langle\ ,\, \rangle$ defined in Section~\ref{section2} is contravariant:
\begin{gather*}
\langle L_{k}v, w \rangle=\langle v, L_{-k}w \rangle, \qquad v,w \in V.
\end{gather*}
Therefore the reduced Fock representation is infinitesimally unitary. The singular vectors are discussed in~\cite{Y}. For the non-reduced Fock representation of the Virasoro algebra, see
for example~\cite{W-Y}.

\begin{Proposition}
\begin{gather*}
L_{-1}q_n=(n+1)q_{n+2}+\frac{1}{2}Q_{n 2}, \qquad n\geq 0.
\end{gather*}
\end{Proposition}
\begin{proof}
It is verified that
\begin{gather*}
 L_{-1}{\rm e}^{\xi(t,u)}
 = \left(\frac{t_1^2}{4}+3t_3\partial_1+5t_5\partial_3+\cdots \right){\rm e}^{\xi(t,u)}\\
\hphantom{L_{-1}{\rm e}^{\xi(t,u)}}{} = \left(\frac{t_1^2}{4}+3t_3u+5t_5u^3+\cdots \right){\rm e}^{\sum_{j\geq 1,\, {\rm odd}}t_ju^j}
\end{gather*}
and
\begin{gather*}
\frac{\partial}{\partial u}{\rm e}^{\xi(t,u)}
 = \frac{\partial}{\partial u}\Big({\rm e}^{\sum_{j\geq 1,\, {\rm odd}}t_ju^j}\Big) \\
 \hphantom{\frac{\partial}{\partial u}{\rm e}^{\xi(t,u)}}{}
 = \big(t_1+3t_3u^2+5t_5u^4+7t_7u^6+\cdots\big){\rm e}^{\sum_{j\geq 1,\, {\rm odd}}t_ju^j}.
\end{gather*}
By the relations $t_1=q_1$ and ${t_1}^2=2q_2$, we have
\begin{gather*}
uL_{-1}{\rm e}^{\xi(t,u)}=\left(\frac{1}{2}uq_2-q_1\right){\rm e}^{\xi(t,u)}+\frac{\partial}{\partial u}{\rm e}^{\xi(t,u)}.
\end{gather*}
Here
\begin{align*}
\left(\frac{1}{2}uq_2-q_1\right)\sum_{n=0}^{\infty}q_nu^n&= \frac{1}{2}q_2\sum_{n=0}^{\infty}q_nu^{n+1}-q_1\sum_{n=0}^{\infty}q_{n}u^n\\
&= u\left(\frac{1}{2}q_2\sum^{\infty}_{n=0}q_nu^n-q_1\sum_{n=0}^{\infty}q_{n+1}u^{n}\right)-q_1,
\end{align*}
and
\begin{gather*}
 \frac{\partial}{\partial u}\sum_{n=0}^{\infty}q_nu^n = \sum_{n=0}^{\infty}nq_nu^{n-1}=u\sum_{n=0}^{\infty}(n+2)q_{n+2}u^n+q_1.
\end{gather*}
Therefore we have
\begin{align*}
 L_{-1}q_n &= \frac{1}{2}q_2q_n-q_1q_{n+1}+(n+2)q_{n+2}\\
 & = (n+1)q_{n+2}+\frac{1}{2}Q_{n 2}.\tag*{\qed}
\end{align*}\renewcommand{\qed}{}
\end{proof}

Similarly, we have
\begin{Proposition}
\begin{gather*}
L_{-2}q_n=(n+2)q_{n+4}+Q_{n 4}-\frac{1}{2}Q_{n 3, 1}, \qquad n\geq0.
\end{gather*}
\end{Proposition}
\begin{proof}
It is verified that
\begin{align*}
L_{-2}{\rm e}^{\xi(t,u)}
&= \left(\frac{3}{2}t_1t_3+5t_5\partial_1+7t_7\partial_3+\cdots \right){\rm e}^{\xi(t,u)}\\
&= \left(\frac{3}{2}t_1t_3+5t_5u+7t_7u^3+\cdots \right){\rm e}^{\sum_{j\geq 1,\, {\rm odd}}t_ju^j}
\end{align*}
and
\begin{align*}
\frac{\partial}{\partial u}{\rm e}^{\xi(t,u)}
&= \frac{\partial}{\partial u}\Big({\rm e}^{\sum_{j\geq 1, \,{\rm odd}}t_ju^j}\Big)\\
&= \big(t_1+3t_3u^2+5t_5u^4+7t_7u^6+\cdots\big){\rm e}^{\sum_{j\geq 1,\, {\rm odd}}t_ju^j}.
\end{align*}
Therefore
\begin{gather*}
 u^3L_{-2}{\rm e}^{\xi(t,u)}=\left(-t_1-3t_3u^2+\frac{3}{2}t_1t_3u^3\right){\rm e}^{\xi(t,u)}+\frac{\partial}{\partial u}{\rm e}^{\xi(t,u)}.
\end{gather*}
We have $t_3=\frac{1}{3}(q_3-Q_{2, 1})$ and $t_1t_3=\frac{1}{3}(2q_4-Q_{3, 1})$.
Therefore
\begin{eqnarray*}
 u^3L_{-2}{\rm e}^{\xi(t,u)}=\left(-q_1-u^2(q_3-Q_{2, 1})+\frac{1}{2}u^3(2q_4-Q_{3, 1})\right){\rm e}^{\xi(t,u)}+\frac{\partial}{\partial u}{\rm e}^{\xi(t,u)}.
\end{eqnarray*}
Here the first term equals
\begin{gather*}
-q_1\sum_{n=0}^{\infty}q_nu^n-(q_3-Q_{2,1})\sum_{n=0}^{\infty}q_nu^{n+2}+\frac{1}{2}(2q_4-Q_{3,1})\sum_{n=0}^{\infty}q_nu^{n+3}\\
\qquad{} =u^3\left(-q_1\sum_{n=0}^{\infty}q_{n+3}u^n-(q_3-Q_{2,1})\sum_{n=0}^{\infty}q_{n+1}u^{n}+\frac{1}{2}(2q_4-Q_{3,1})\sum_{n=0}^{\infty}q_nu^n\right)\\
\qquad\quad{} +u^3\big({-}q_1\big(u^{-3}+q_1u^{-2}+q_{2}u^{-1}\big)-(q_3-Q_{2,1})u^{-1}\big)
\end{gather*}
and the second term equals
\begin{gather*}
 \sum_{n=0}^{\infty}nq_nu^{n-1}=u^3\sum_{n=0}^{\infty}(n+4)q_{n+4}u^n+u^3\big(q_1u^{-3}+2q_{2}u^{-2}+3q_{3}u^{-1}\big).
\end{gather*}
Also it is easy to check that
\begin{gather*}
u^3\big({-}q_1\big(u^{-3}+q_1u^{-2}+q_{2}u^{-1}\big)-(q_3-Q_{2,1})u^{-1}\big)+u^3\big(q_1u^{-3}+2q_{2}u^{-2}+3q_{3}u^{-1}\big)\\
\qquad{} =u^2(2q_3-q_1q_2+Q_{2,1})=0.
\end{gather*}
Hence
\begin{align*}
L_{-2}q_n& = -q_{1}q_{n+3}-(q_3-Q_{2,1})q_{n+1}+\frac{1}{2}(2q_4-Q_{3,1})q_n+(n+4)q_{n+4} \\
&=-q_{1}q_{n+3}-(3q_3-q_2q_1)q_{n+1}+\frac{1}{2}(4q_4-q_{3}q_{1})q_n+(n+4)q_{n+4} \\ 
& =(n+2)q_{n+4}+Q_{n4}-\frac{1}{2}Q_{n3,1}.\tag*{\qed} 
\end{align*}\renewcommand{\qed}{}
\end{proof}

From Proposition \ref{Q-relation1}, we obtain the following formula.
\begin{align*}
L_{-k}(vw)&= (L_{-k}v)w+v(L_{-k}w)-\frac{1}{4}\sum_{j=0}^{k-1}(2j+1)t_{2j+1}(2k-(2j+1))t_{2k-(2j+1)}vw\\
 &= (L_{-k}v)w+v(L_{-k}w)-\frac{1}{2}\sum_{j=0}^{k-1}{(-1)}^{j}(k-j)Q_{2k-j,j}vw
\end{align*}
for $v,w \in V$. In particular $k=1, 2$, we see
\begin{gather}
L_{-1}(vw)=(L_{-1}v)w+v(L_{-1}w)-\frac{1}{2}q_2vw,\label{LR_{-1}}\\
L_{-2}(vw)=(L_{-2}v)w+v(L_{-2}w)-\frac{1}{2}(2q_4-Q_{3,1})vw.\label{LR_{-2}}
\end{gather}

\begin{Proposition}\label{LQ}
Let $\alpha=(\alpha_1,\alpha_2,\dots, \alpha_{\ell})$ be a positive integer sequence. Then
\begin{gather*}
(1)\quad L_{-1}Q_{\alpha}=\sum_{i=1}^{\ell}(\alpha_i+1)Q_{\alpha+2\epsilon_i}+\frac{1}{2}Q_{\alpha,2},\\
(2)\quad L_{-2}Q_{\alpha}=\sum_{i=1}^{\ell}(\alpha_i+2)Q_{\alpha+4\epsilon_i}+Q_{\alpha,4}-\frac{1}{2}Q_{\alpha,3,1}.
\end{gather*}
\end{Proposition}
\begin{proof}
If $\alpha_i=\alpha_j$ for some $i \ne j$, then the equations hold as $0=0$. Therefore, taking the sign~$(\pm 1)$ into account, it suffices to prove the equations for the case $\alpha=\lambda$ is a strict partition. Use induction on the length of $\lambda$. First we see~(1) for the case $\ell(\lambda)$ is odd. By equation~(\ref{LR_{-1}}),
\begin{gather}
 L_{-1}\left(\sum_{i=1}^{\ell}{(-1)}^{i+1}q_{\lambda_i}Q_{\lambda_1 \dots \widehat{\lambda_i}\dots \lambda_{\ell}}\right)\nonumber\\
\qquad{} =\sum_{i=1}^{\ell}{(-1)}^{i+1}\left((L_{-1}q_{\lambda_i})Q_{\lambda_1 \dots \widehat{\lambda_i}\dots \lambda_{\ell}}+q_{\lambda_i}\big(L_{-1}Q_{\lambda_1 \dots \widehat{\lambda_i}\dots \lambda_{\ell}}\big)-\frac{1}{2}q_{2}q_{\lambda_i}Q_{\lambda_1 \dots \widehat{\lambda_i}\dots \lambda_{\ell}}\right).\!\!\!\!\label{L_{-1}Q}
\end{gather}
By induction hypothesis, and the first term and second term in the right hand side equal, respectively,
\begin{gather*}
 \sum_{i=1}^{\ell}{(-1)}^{i+1}\left((\lambda_i+1)q_{\lambda_i+2}+\frac{1}{2}Q_{\lambda_i 2}\right)Q_{\lambda_1 \dots \widehat{\lambda_i}\dots \lambda_{\ell}},\qquad \text{and}\\
 \sum_{i=1}^{\ell}{(-1)}^{i+1}q_{\lambda_i}\left(\sum_{j=1,\,j\ne i}^{\ell}(\lambda_j+1)Q_{\lambda_1 \dots \widehat{\lambda_i}\dots \lambda_{\ell} +2\epsilon_j}+\frac{1}{2}Q_{\lambda_1 \dots \widehat{\lambda_i}\dots \lambda_{\ell} 2}\right).
\end{gather*}
Hence the equation (\ref{L_{-1}Q}) reads
\begin{gather*}
\sum_{i=1}^{\ell}{(-1)}^{i+1}\left((\lambda_i+1)q_{\lambda_i+2}Q_{\lambda_1 \dots \widehat{\lambda_i}\dots \lambda_{\ell}}+q_{\lambda_i}\sum_{j=1,j\ne i}^{\ell}(\lambda_j+1)Q_{\lambda_1 \dots \widehat{\lambda_i}\dots \lambda_{\ell} +2\epsilon_j}\right)\\
 \qquad\quad{} +\frac{1}{2}\sum_{i=1}^{\ell}{(-1)}^{i+1}\big(Q_{\lambda_i 2}Q_{\lambda_1 \dots \widehat{\lambda_i}\dots \lambda_{\ell}}+q_{\lambda_i}Q_{\lambda_1 \dots \widehat{\lambda_i}\dots \lambda_{\ell} 2}-q_2q_{\lambda_i}Q_{\lambda_1 \dots \widehat{\lambda_i}\dots \lambda_{\ell}} \big)\\
\qquad{} =\sum_{i=1}^{\ell}(\lambda_i+1)Q_{\lambda+2\epsilon_i}+\frac{1}{2}\sum_{i=1}^{\ell}{(-1)}^{i+1}\big(Q_{\lambda_i 2}Q_{\lambda_1 \dots \widehat{\lambda_i}\dots \lambda_{\ell}}+q_{\lambda_i}Q_{\lambda_1 \dots \widehat{\lambda_i}\dots \lambda_{\ell} 2}\big) -\frac{1}{2}q_2Q_{\lambda}.
\end{gather*}
By Lemma \ref{Q-relation2}(1), the result follows. The case of even $\ell(\lambda)$ is similar.
Next we prove (2) in Proposition~\ref{LQ}. Let $\ell(\lambda)$ be odd. By equation~(\ref{LR_{-2}}),
\begin{gather}
 L_{-2}\left(\sum_{i=1}^{\ell}{(-1)}^{i+1}q_{\lambda_i}Q_{\lambda_1 \dots \widehat{\lambda_i}\dots \lambda_{\ell}}\right)\label{L_{-2}Q}\\
 =\sum_{i=1}^{\ell}{(-1)}^{i+1}\left((L_{-2}q_{\lambda_i})Q_{\lambda_1 \dots \widehat{\lambda_i}\dots \lambda_{\ell}}+q_{\lambda_i}\big(L_{-2}Q_{\lambda_1 \dots \widehat{\lambda_i}\dots \lambda_{\ell}}\big)-\frac{1}{2}(2q_{4}-Q_{3, 1})q_{\lambda_i}Q_{\lambda_1 \dots \widehat{\lambda_i}\dots \lambda_{\ell}}\right).\nonumber
\end{gather}
By the induction hypothesis, the first and second terms in the right hand side are, respectively,
\begin{gather*}
 \sum_{i=1}^{\ell}{(-1)}^{i+1}\left((\lambda_i+2)q_{\lambda_i+4}+Q_{\lambda_i 4}-\frac{1}{2}Q_{\lambda_i 3, 1}\right)Q_{\lambda_1 \dots \widehat{\lambda_i}\dots \lambda_{\ell}},\qquad \text{and}\\
 \sum_{i=1}^{\ell}{(-1)}^{i+1}q_{\lambda_i}\left(\sum_{j=1,j\ne i}^{\ell}(\lambda_j+2)Q_{\lambda_1 \dots \widehat{\lambda_i}\dots \lambda_{\ell} +4\epsilon_j}+Q_{\lambda_1 \dots \widehat{\lambda_i}\dots \lambda_{\ell} 4}-\frac{1}{2}Q_{\lambda_1 \dots \widehat{\lambda_i}\dots \lambda_{\ell} 3, 1}\right).
\end{gather*}
Hence the equation (\ref{L_{-2}Q}) reads
\begin{gather*}
 \sum_{i=1}^{\ell}{(-1)}^{i+1}\left((\lambda_i+2)q_{\lambda_i+4}Q_{\lambda_1 \dots \widehat{\lambda_i}\dots \lambda_{\ell}}+q_{\lambda_i}\sum_{j=1,j\ne i}^{\ell}(\lambda_j+2)Q_{\lambda_1 \dots \widehat{\lambda_i}\dots \lambda_{\ell} +4\epsilon_j}\right)\\
\qquad{}+\sum_{i=1}^{\ell}{(-1)}^{i+1}\big(Q_{\lambda_i 4}Q_{\lambda_1 \dots \widehat{\lambda_i}\dots \lambda_{\ell}}+q_{\lambda_i}Q_{\lambda_1 \dots \widehat{\lambda_i}\dots \lambda_{\ell} 4}-q_4q_{\lambda_i}Q_{\lambda_1 \dots \widehat{\lambda_i}\dots \lambda_{\ell}}\big)\\
\qquad{} +\frac{1}{2}\sum_{i=1}^{\ell}{(-1)}^{i+1}\big({-}Q_{\lambda_i 3, 1}Q_{\lambda_1 \dots \widehat{\lambda_i}\dots \lambda_{\ell}}-q_{\lambda_i}Q_{\lambda_1 \dots \widehat{\lambda_i}\dots \lambda_{\ell} 3, 1}+Q_{3, 1}q_{\lambda_i}Q_{\lambda_1 \dots \widehat{\lambda_i}\dots \lambda_{\ell}}\big).
\end{gather*}
Hence
\begin{gather*}
\sum_{i=1}^{\ell}(\lambda_i+2)Q_{\lambda+4\epsilon_i}+\sum_{i=1}^{\ell}{(-1)}^{i+1}\big(Q_{\lambda_i 4}Q_{\lambda_1 \dots \widehat{\lambda_i}\dots \lambda_{\ell}}+q_{\lambda_i}Q_{\lambda_1 \dots \widehat{\lambda_i}\dots \lambda_{\ell} 4}\big)-q_4Q_{\lambda}\\
\qquad{} +\frac{1}{2}\sum_{i=1}^{\ell}{(-1)}^{i+1}\big({-}Q_{\lambda_i 3, 1}Q_{\lambda_1 \dots \widehat{\lambda_i}\dots \lambda_{\ell}}-q_{\lambda_i}Q_{\lambda_1 \dots \widehat{\lambda_i}\dots \lambda_{\ell} 3, 1}\big)+\frac{1}{2}Q_{3, 1}Q_{\lambda}.
\end{gather*}
The result follows immediately from Lemma \ref{Q-relation2}(1) and (3). The case of even $\ell(\lambda)$ is similar.
\end{proof}

\begin{Theorem}\label{TM1}
Let $\alpha=(\alpha_1,\alpha_2,\dots,\alpha_{\ell})$ be a positive integer sequence. Then
\begin{gather*}
L_{-k}Q_{\alpha}=\sum_{i=1}^{\ell}(\alpha_i+k)Q_{\alpha+2k\epsilon_i}+\frac{1}{2}\sum_{i=0}^{k-1}{(-1)}^{i}(k-i)Q_{\alpha,2k-i,i}, \qquad k\geq1.
\end{gather*}
\end{Theorem}
\begin{proof}
Use induction on $k$. The cases $k=1,2$ are already shown in Proposition~\ref{LQ}. Thanks to the Virasoro relations, it suffices to show
\begin{gather*}
[L_{-k},L_{-1}]Q_{\alpha} \\
\qquad {} = 2(-k+1)\left(\sum_{i=1}^{\ell}(\alpha_i+(k+1))Q_{\alpha+2(k+1)\epsilon_i}
+\frac{1}{2}\sum_{i=0}^{k}{(-1)}^{i}((k+1)-i)Q_{\alpha, 2(k+1)-i, i}\right).
\end{gather*}
Since
\begin{gather*}
L_{-k}L_{-1}Q_{\alpha}=\sum_{i=1}^{\ell}(\alpha_i+1)L_{-k}Q_{\alpha+2\epsilon_{i}}+\frac{1}{2}L_{-k}Q_{\alpha 2}, \qquad \text{and}\\
L_{-1}L_{-k}Q_{\alpha}=\sum_{i=1}^{\ell}(\alpha_i+k)L_{-1}Q_{\alpha+2k\epsilon_{i}}+\frac{1}{2}\sum_{i=0}^{k-1}{(-1)}^{i}(k-i)L_{-1}Q_{\alpha, 2k-i, i},
\end{gather*}
we have
\begin{gather*}
 [L_{-k}, L_{-1}]Q_{\alpha}=\sum_{i=1}^{\ell}\big((\alpha_i+1)L_{-k}Q_{\alpha+2\epsilon_i}-(\alpha_i+k)L_{-1}Q_{\alpha+2k\epsilon_i}\big)\\
\hphantom{[L_{-k}, L_{-1}]Q_{\alpha}=}{} +\frac{1}{2}\left(L_{-k}Q_{\alpha 2}-\sum_{i=0}^{k-1}{(-1)}^{i}(k-i)L_{-1}Q_{\alpha, 2k-i, i}\right).
\end{gather*}
We write down terms in the right hand side:
\begin{gather*}
\sum_{i=1}^{\ell}(\alpha_i+1)L_{-k}Q_{\alpha+2\epsilon_{i}}=\sum_{i,j=1,i\ne j}^{\ell}(\alpha_i+1)(\alpha_j+k)Q_{\alpha+2\epsilon_i+2k\epsilon_j}\\
\qquad{} +\sum_{i=1}^{\ell}(\alpha_i+1)(\alpha_i+k+2)Q_{\alpha+2(k+1)\epsilon_i}
+\sum_{i=1}^{\ell}(\alpha_{i}+1)\frac{1}{2}\sum_{j=0}^{k-1}{(-1)}^{j}(k-j)Q_{\alpha+2\epsilon_i, 2k-j, j},\\
\sum_{i=1}^{\ell}(\alpha_i+k)L_{-1}Q_{\alpha+2k\epsilon_{i}}\\
\quad{} =\sum_{i=1}^{\ell}(\alpha_i+k)\left(\sum^{\ell}_{j=1, j\ne i}(\alpha_j+1)Q_{\alpha+2k\epsilon_i+2\epsilon_j}
+(\alpha_i+2k+1)Q_{\alpha+2(k+1)\epsilon_{i}}+\frac{1}{2}Q_{\alpha+2k\epsilon_i, 2}\right),\\
L_{-k}Q_{\alpha, 2}=\sum_{i=1}^{\ell}(\alpha_i+k)Q_{\alpha+2k\epsilon_i, 2}+(2+k)Q_{\alpha, 2+2k}
+\frac{1}{2}\sum_{i=0}^{k-1}{(-1)}^{i}(k-i)Q_{\alpha, 2, 2k-i, i},\\
\sum_{i=0}^{k-1}{(-1)}^{i}(k-i)L_{-1}Q_{\alpha, 2k-i, i}=\sum_{i=0}^{k-1}{(-1)}^{i}(k-i)\sum_{j=1}^{\ell}(\alpha_j+1)Q_{\alpha+2\epsilon_j, 2k-i, i}\\
\qquad{} +\sum_{i=0}^{k-1}{(-1)}^{i}(k-i)(2k-i+1)Q_{\alpha, 2k-i+2, i}
+\sum_{i=0}^{k-1}{(-1)}^{i}(k-i)(i+1)Q_{\alpha, 2k-i, i+2}\\
\qquad{} -\frac{1}{2}kQ_{\alpha, 2k, 2}+\frac{1}{2}\sum_{i=1}^{k-1}{(-1)}^{i}(k-i)Q_{\alpha, 2k-i, i, 2}.
\end{gather*}
Summing up, we have
\begin{gather*}
\left(\sum_{i=1}^{\ell}(\alpha_i+1)(\alpha_i+k+2)-\sum_{i=1}^{\ell}(\alpha_i+k)(\alpha_i+2k+1)\right)Q_{\alpha+2(k+1)\epsilon_{i}}\\
\qquad\quad{} +\frac{1}{2}\Biggl((2+k)Q_{\alpha, 2+2k}-\sum_{i=0}^{k-1}{(-1)}^{i}(k-i)(2k-i+1)Q_{\alpha, 2k-i+2, i}\\
\qquad\quad {} +\frac{1}{2}\sum_{i=0}^{k-1}{(-1)}^{i}(k-i)Q_{\alpha, 2, 2k-i, i}+\frac{1}{2}kQ_{\alpha, 2k, 2}-\frac{1}{2}\sum_{i=1}^{k-1}{(-1)}^{i}(k-i)Q_{\alpha, 2k-i, i, 2}\\
\qquad\quad{} -\sum_{i=0}^{k-1}{(-1)}^{i}(k-i)(i+1)Q_{\alpha, 2k-i, i+2}\Biggl)
\\
\qquad{}=2(-k+1)\sum_{i=1}^{\ell}(\alpha_i+(k+1))Q_{\alpha+2(k+1)\epsilon_i}+(-k+1)(k+1)Q_{\alpha, 2+2k}\\
\qquad\quad{}-\frac{1}{2}\Biggl( \sum_{i=1}^{k-1}{(-1)}^{i}(k-i)(2k-i+1)Q_{\alpha, 2k-i+2, i}
+kQ_{\alpha, 2k, 2}\\
\qquad\quad{} +\sum_{i=1}^{k-1}{(-1)}^{i}(k-i)(i+1)Q_{\alpha, 2k-i, i+2}\Biggl)\\
\qquad {}=2(-k+1)\sum_{i=1}^{\ell}(\alpha_i+(k+1))Q_{\alpha+2(k+1)\epsilon_i}+(-k+1)(k+1)Q_{\alpha, 2+2k}\\
\qquad\quad{} +(-k+1)\sum_{i=1}^{k}{(-1)}^{i}((k+1)-i)Q_{\alpha, 2(k+1)-i, i}.\tag*{\qed}
\end{gather*}\renewcommand{\qed}{}
\end{proof}

Together with the previously proved formula
\begin{eqnarray*}
L_{k}Q_{\lambda} = \sum_{i=1}^{2m}(\lambda_{i}-k) Q_{\lambda-2k\epsilon_i}
\end{eqnarray*}
for $\lambda = (\lambda_{1},\dots, \lambda_{2m})$, $k\geq1$ \cite[Theorem~2]{A-S-Y}, Theorem \ref{TM1}
completely describes the reduced Fock representation of the Virasoro algebra.
Consider the Lie subalgebra
$\mathfrak{g} = \sum_{|k|\leq1}\mathbb{C}l_{k}$ which is isomorphic to $\mathfrak{sl}(2,\mathbb{C})$.
Let $\mathcal{ESP}$ be the set of the strict partitions whose parts are all even numbers, and let
$V^{\rm even}$ be the subspace of $V$ spanned by the $Q_{\lambda}$ for $\lambda \in \mathcal{ESP}$.

\begin{Corollary}
The space $V^{\rm even}$ is invariant under the action of $\mathfrak{g}$.
\end{Corollary}

We are interested in the space $V^{\rm even}$ because of the following conjecture: the set of Hirota bilinear equations
\begin{eqnarray*}
Q_{\lambda}\big(\widetilde{D}\big)\tau \cdot \tau=0,\qquad \lambda \in \mathcal{SP}\backslash \mathcal{ESP}
\end{eqnarray*}
coincides with those of the KdV hierarchy, where $\widetilde{D}=\big(D_1,\frac{1}{3}D_3,\frac{1}{5}D_5,\dots\big)$ is the Hirota differential operator. Namely $\big(V^{\rm even}\big)^{\perp}$ is conjecturally the space of Hirota equations for the KdV hierarchy.
For example,
\begin{eqnarray*}
Q_{3,1}\big(\widetilde{D}\big)=\frac{1}{12}D_1^4-\frac{1}{3}D_1D_3
\end{eqnarray*}
corresponds to the original KdV equation.

\subsection*{Acknowledgements}

After completing the draft of the present paper we had a chance to see the preprint~\cite{L-Y}, where the same formula as ours is proved as a part of authors' theory on BGW tau functions. Since the proof is slightly different, we decided to keep our proof in the present paper. We are grateful to the referees for informing~\cite{L-Y} to us, as well as other useful comments which improved the paper.
The funding was provided by KAKENHI (Grant No.~17K05180).

\pdfbookmark[1]{References}{ref}
\LastPageEnding

\end{document}